\documentclass[a4paper,UKenglish]{llncs}
\usepackage{amsmath}
\usepackage{amssymb}
\usepackage{pgf}
\usepackage{marvosym}
\usepackage{tikz}
\usetikzlibrary{arrows,automata}
\usepackage[T1]{fontenc}
\usepackage[shortlabels]{enumitem}

\title{$\Sigma^{\mu}_2$ is decidable for $\Pi^{\mu}_2$}
\subtitle{(Pre-print, extended version)}

\author{Karoliina Lehtinen\inst{1}\inst{3}\textsuperscript{\Letter} \and Sandra Quickert\inst{2}\inst{3}}
\institute{University of Kiel, Germany \\
\email{kleh@informatik.uni-kiel.de}
\and University of St. Andrews, UK \\
\email{sq21@st-andrews.ac.uk}
\and University of Edinburgh, UK}

\begin{document}

\maketitle

\begin{abstract}
Given a $\Pi^{\mu}_2$ formula of the modal $\mu$ calculus, it is decidable whether it is equivalent to a $\Sigma^{\mu}_2$ formula.
 \end{abstract}

 \section{Introduction}
 
 The modal $\mu$ calculus, $L_\mu$, is a well-established verification logic describing properties of labelled transition systems.
 It consists of a simple modal logic, augmented with the least fixpoint $\mu$ and its dual, the greatest fixpoint $\nu$.
 Alternations between $\mu$ and $\nu$ are key for measuring complexity: the fewer alternations, the easier a formula
 is to model check. We call this the formula's index. For any fixed index, the model-checking problem is in $P$. However,
 no fixed index is sufficient to capture all properties expressible in $L_\mu$ \cite{bradfield1996modal,lenzi1996hierarchy,arnold1999}, and it is notoriously
 difficult to decide whether a formula can be simplified. So far only properties expressible without fixpoints \cite{otto1999eliminating}, or with only one type of fixpoint \cite{kusters2002deciding} are known to be decidable.
 
 In automata theory, the corresponding index problem is to decide the simplest acceptance condition sufficient to express a property with a
 specified type of automata. This is often referred to as the Mostowski-Rabin index of a language.
 
Given a deterministic automaton on labelled binary trees,
 the minimal index of equivalent deterministic \cite{niwinski1998relating}, non-deterministic \cite{urbanski2000deciding,niwinski2005deciding},
 and alternating \cite{niwinski2003gap} automata are all known to be decidable.
In \cite{facchini2013rabin} these results were extended to show that the non-deterministic and alternating index problems are also decidable for languages 
of labelled binary trees recognised by game-automata, a slightly more general model than deterministic automata.

 For the case of non-deterministic automata, the index problem reduces to the uniform universality of distance parity automata \cite{colcombet2008non}.
In \cite{colcombet2013deciding} it was shown that given a B\"uchi definable language $\mathcal L$, it is decidable whether it can be described by an alternating co-B\"uchi automaton. 
 Skrzypczak and Walukiewicz \cite{skrzypczakdeciding} give an alternative proof of the same result and add a topological 
 characterisation of the recognised languages. 
 A B\"uchi definable language which is co-B\"uchi is said to be weakly definable: it is definable in weak monadic second order logic \cite{rabin1970}, and equivalently, by an alternating automaton which is simultaneously
both B\"uchi and co-B\"uchi.
In $L_\mu$ terms, this result corresponds to deciding whether
 a formula in the class $\Pi^{\mu}_2$ is equivalent on binary trees to a formula in the class $\Sigma^{\mu}_2$.

 This paper provides a novel proof of the same result extended to arbitrary structures: given a $\Pi^{\mu}_2$ formula, it is decidable whether it is equivalent to a $\Sigma^{\mu}_2$ formula.
 The proof defines an $n$-parametrised game such that the decidability of $\Sigma^{\mu}_2$ reduces to deciding whether for some $n$ this is the model-checking game for a formula. 
 From this game, we derive a family $\Psi^n$ of $\Sigma^{\mu}_2$ formulas, such that an input formula $\Psi$ is equivalent to a $\Sigma^{\mu}_2$ formula if and only if it is
 equivalent to some formula in this family. To decide the parameter $n$ for $\Pi^{\mu}_2$ input formulas, we simply argue that the game construction in
 \cite{skrzypczakdeciding}, designed for binary trees, extends to the case of labelled transition systems.
 
 We consider the most interesting contributions of this paper to be the reduction of the decidability of $\Sigma^{\mu}_2$ to finding the parameter $n$ such that $\Psi$ is equivalent to $\Psi^n$.
 With this result, finding a way to generalise the game construction from \cite{skrzypczakdeciding} to arbitrary inputs would suffice to decide $\Sigma^{\mu}_2$.

 \section{Preliminaries}\label{sec-preliminaries}
 
 \subsection{$L_\mu$}
 Let us fix, once and for all a finite set of actions $\mathit{Act}=\{a,b,...\}$, a countably infinite set of propositional variables $\mathit{Prop}=\{P,Q,...\}$, and
 fixpoint variables $\mathit{Var}=\{X,Y,...\}$.  A literal is either $P$ or $\neg P$ for $P \in \mathit{Prop}$.

  \begin{definition} {\em (Labelled transition tree)}
  A labelled transition tree is a structure $\mathcal{T}=(V,v_r,E,L,P)$ where $V$ is a set of states, $v_r$ is the root, the only node without predecessor, $E\subset V\times V$ is an edge relation,
  $L: E \rightarrow \mathit{Act}$ labels edges with actions and $P: V\rightarrow 2^{\mathit{Prop}}$ labels vertices with propositional variables. Furthermore, for each $v\in V$ the set of ancestors
  $\{w\in V\mid \exists  w_1,\ldots,w_k.\;wEw_1E\ldots w_kEv\}$ is finite and well-ordered with respect to the transitive closure of $E$; the set of successors $\{w\in V\mid vEw\}$ is also finite.

  We can represent repetition in an infinite tree with back edges. Note that we allow more than one successor per label.
 \end{definition}

 \begin{definition}
 (\textit{Modal} $\mu$) The syntax of
$L_{\mu}$ is given by:
\[ \phi:= P 
	\mbox{ | } X
   \mbox{ | }\neg P
   \mbox{ | }\phi\wedge\phi
   \mbox{ | }\phi\vee\phi
   \mbox{ | }\langle a \rangle \phi
   \mbox{ | }[a] \phi
   \mbox{ | }\mu X.\phi
   \mbox{ | }\nu X.\phi
   \mbox{ | }\bot
   \mbox{ | }\top
\]

 \end{definition}
The order of operator precedence is $[a], \langle a \rangle,
\wedge, \vee, \mu$ and $\nu$.

The operators $\langle a \rangle$ and $[a]$ are called {\em modalities}, and formulas $\langle a \rangle \phi$ and $[a]\phi$ are called {\em modal formulas}. If $\psi=\mu X.\phi$ or $\psi=\nu X.\phi$, we call the formula $\phi$ the binding formula of $X$ within $\psi$ and denote it by $\phi_X$. 
We say that $\phi'$ is an immediate subformula of $\phi$ if either $\phi$ is built from $\phi'$ in one step using the syntax rules above, or, in a slight abuse of notation, if $\phi=X$ and $\phi'$ its binding formula.
 Hence $\phi$ is an immediate subformula of the formulas $\phi\vee \psi$, $\langle a \rangle \phi$, $\mu X.\phi$ and also of $X$ in $\nu X.\phi$.
A formula is guarded if every fixpoint variable is in the scope of a modality within its
binding. Without loss of expressivity 
\cite{MateescuRadu2002,kupferman2000automata}, we
restrict ourselves to $L_\mu$ in guarded positive form. We will also assume throughout the paper that all fixpoint variables within a formula have distinct names.

 The semantics of $L_\mu$ are standard, see for example \cite{bradfield2007modal}.
  We now define the {\em priority assignment} and {\em index} of a formula, following Niwi\'nski's notion of alternation in \cite{niwinski86}.
 
 \begin{definition}{\em (Priority assignment, index and alternation classes)}
  A {\em priority assignment} $\Omega$ is a function assigning an integer value to each fixpoint variable in a formula such that: 
  (a) $\mu$-bound variables receive odd priorities and $\nu$-bound variables receive even priorities, and (b) if $X$ is free in $\phi_Y$, the binding formula of $Y$, then $\Omega(X)\geq \Omega(Y)$.
A formula has {\em index} $\{q,...,i\}$ where $i\in \{0,1\}$ if it has a priority assignment with co-domain $\{q,...,i\}$.

 Formulas without fixpoints form the modal fragment of $L_\mu$. Formulas with one type of fixpoint have index $\{0\}$ or $\{1\}$,
 corresponding to the alternation classes $\Pi^{\mu}_1$ and $\Sigma^{\mu}_1$, respectively. Then the class $\Pi^{\mu}_i$ and $\Sigma^{\mu}_i$ for even
 $i$ correspond to formulas with indices $\{i,...,1\}$ and $\{i-1,...,0\}$, respectively, while for odd $i$ they correspond
 formulas with indices $\{i-1,...,0\}$ and $\{i,...,1\}$, respectively.
 A formula has semantic alternation class $C$ if it is equivalent to a formula in $C$.
  \end{definition}
 
 \begin{example}
 The formula $\mu X. \nu Y. \Box Y \wedge \mu Z. \Box (X \vee Z)$ accepts the priority assignment $\Omega(X)=1,\Omega(Y)=0$ and $\Omega(Z)=1$, so it has index $\{1,0\}$ and is in the class $\Sigma^{\mu}_2$. However, it is equivalent to $\mu X. \Box X$
which holds in structures without infinite paths, and is therefore semantically in $\Sigma^{\mu}_1$.
\end{example}

In this paper we present a new proof that $\Sigma^{\mu}_2$ is decidable for formulas in $\Pi^{\mu}_2$: given an arbitrary $L_\mu$ formula $\Psi$ with index $\{2,1\}$, it is decidable whether
$\Psi$ is equivalent to a formula with index $\{1,0\}$.

 \subsection{Parity Games}
 
 The semantics of $L_\mu$ formulas (like that of alternating parity automata) can be described in terms of winning regions of parity games.
 
 \begin{definition}
  A parity game $G=(V,v_i,E,\Omega)$ consists of a set of vertices $V$ partitioned into those belonging to Even, $V_e$,
  and those belonging to Odd, $V_o$, an initial position $v_i \in V$, and a set of edges $E\in V\times V$. A priority assignment $\Omega$ assigns a priority to every vertex.
  
  At each turn, the player who owns the current position $v$ chooses a successor
  position from the successors of $v$ via $E$. A play is a potentially infinite sequence of positions starting at the initial position $v_i$. A finite play is winning for Even if the final position
  has even priority, and for Odd otherwise. An infinite play is winning for the player of the parity of the highest priority seen infinitely often.
 \end{definition}

 Parity games are known to be determined and we can restrict ourselves to positional winning strategies \cite{emersonjutla91},\cite{mostowski91}.
 It is a standard result that given a structure $\mathcal{M}$ and a formula $\Psi$, there is a model-checking parity game $\mathcal{M}\times \Psi$ such that Even wins if
 and only if $\mathcal{M}$ satisfies $\Psi$ \cite{wilke2001alternating}.
 
 \begin{definition}{\em (The model-checking game $\mathcal{M}\times \Psi$)}
  The parity game $\mathcal{M}\times \Psi$ has for states $s \times \phi$ where $s$ is a state of $\mathcal{M}$
  and $\phi$ is a subformula of $\Psi$. There is an edge from $s \times \psi$ to $s \times \phi$ if
  $\phi$ is an immediate subformula of a non-modal formula $\psi$; there is an edge from $s \times \langle a \rangle\phi$ and $s \times [a]\phi$
  to $s'\times \phi$ for $s'$ an $a$-successor of $s$.
  Positions $s\times \phi$ where $\phi$ is a disjunction or starting with an existential modality $\langle a \rangle$ belong to Even while
  those where $\phi$ is a conjunction or universal modality $[a]$ belong to Odd. Positions with a single successor are given to Even,
  although the game is deterministic at those.
  The priority assignment is inherited from the priority assignment $\Omega_{\Psi}$ on $\Psi$: a fixpoint variable $X$ receives priority $\Omega_{\Psi}(X)$ while
  other nodes receive the minimal priority in the co-domain of $\Omega_{\Psi}$.
 \end{definition}
 
 \subsection{Disjunctive Form}
 
 Disjunctive $L_\mu$ is a fragment restricting conjunctions in a
 way reminiscent of non-deterministic automata \cite{Walukiewicz2000}. Its use is key to several of our proofs.
 
 \begin{definition}
\textit{(Disjunctive formulas) }The set of disjunctive   form formulas
of  $L_{\mu}$ is the smallest set $\mathcal{F}$ satisfying:
\begin{itemize}
\item $\top$,$\bot$, fixpoint variables and finite sets (conjunctions) of literals are in $\mathcal{F}$;
\item If $\psi\in\mathcal{F}$ and $\phi\in\mathcal{F}$, then $\psi\vee\phi\in\mathcal{F}$; 
\item If for each $a$ in $\mathit{Act}$ the set $\mathcal{B}_a\subseteq\mathcal{F}$
is a finite set of formulas, and if $\mathcal{A}$ is a finite set of literals, 
then $\mathcal{A}\wedge \bigwedge_{a\in \mathit{Act}}   {\xrightarrow{a}} \mathcal{B}_a\in \mathcal{F}$
where ${\xrightarrow{a}}\mathcal{B}_a$ is short for $(\bigwedge_{\psi\in\mathcal{B}_a}\langle a \rangle \psi)\wedge [a] \bigvee_{\psi\in\mathcal{B}_a}\psi$
-- that is to say, every formula in $\mathcal{B}_a$ holds
at least one successor and at every successor at least one of
the formulas in $\mathcal{B}_a$ holds;
\item $\mu X.\psi$ and $\nu X.\psi$ are in $\mathcal{F}$ as long as $\psi\in\mathcal{F}$.
\end{itemize}
 \end{definition}
Every formula is known to be equivalent to an effectively computable
formula in  disjunctive form \cite{Walukiewicz2000}. The transformation preserves
guardedness.
 
Given an $L_\mu$ formula with unrestricted conjunctions,
 the model-checking parity game requires Even to have a strategy to verify both conjuncts. A strategy for Even will
 agree with the plays corresponding to each of Odd's choices, leading potentially to several plays on some branches.
In contrast, disjunctive form restricts conjunctions, and the only branching in Even's strategies is at a position where the formula is of the form $\mathcal{A} \wedge \bigwedge_{a\in \mathit{Act}}  {\xrightarrow{a}} \mathcal{B}_a\in \mathcal{F}$, called an {\em Odd-choice formula}.

Disjunctive form guarantees that Even can use strategies which only agree with one play per branch. For further details, see \cite{lehtinen2015deciding}.
 
 \begin{lemma}\cite{lehtinen2015deciding}\label{below}
  Given a disjunctive formula $\Psi$, for any structure $\mathcal{M}$ and strategy $\sigma$ in $\mathcal{M}\times \Psi$, there is a structure $\mathcal{M}'$ bisimilar to $\mathcal{M}$
  such that a strategy $\sigma'$ in $\mathcal{M}'\times \Psi$ induced from $\sigma$ only agrees with one play per branch. 
  We then say that $\mathcal{M}'$ and $\sigma'$ are well-behaved.
 \end{lemma}

\begin{lemma}\label{disjunctpitwo}
 Given a $\Pi^{\mu}_2$ formula, the transformation into disjunctive form as presented in \cite{Walukiewicz2000} yields a disjunctive $\Pi^{\mu}_2$
 formula.
\end{lemma}

The proof, in Appendix \ref{app-preserve-pi}, uses the concepts of tableau, tableau equivalence, and traces from \cite{Walukiewicz2000}. The crux of the argument
is that the tableau of a disjunctive formula not in $\Pi^{\mu}_2$ must have an even cycle nested in an odd cycle which in turn implies
the existence of a trace on which a $\mu$-fixpoint dominates a $\nu$ fixpoint in any equivalent tableau.

Note that the dual is not true:
a $\Sigma^{\mu}_2$ formula may yield a formula of arbitrarily large alternation depth when turned into disjunctive form \cite{lehtinendisjunctive}.
This is in line with alternating B\"uchi automata being equivalent to non-deterministic B\"uchi automata while the same is not
necessarily true for co-B\"uchi automata.

 \subsection{Automata and $L_{\mu}$}\label{sec-non-det-automata}

 The relationship between $L_\mu$ and automata theory is based on the fact that the automata model that $L_\mu$ formulas correspond to is, when restricted to 
 binary trees\footnote{assuming $|\mathit{Act}|=2$; otherwise trees with one successor per label.}, equivalent to alternating automata with a parity condition \cite{janin1995automata}.
 The model-checking problems in these two settings are equivalent:
 Model-checking a formula $\psi$ on a structure $\mathcal{M}$ reduces to checking an automaton $A(\psi)$ on a binary tree encoding of $\mathcal{M}$. Model checking disjunctive $L_\mu$ similarly reduces to model checking non-deterministic automata, albeit one of potentially higher index. For the index problem, the comparison is not as simple and to the best of our knowledge there is no known reduction from the (disjunctive) $L_\mu$ index problem to the (non-deterministic) automata index problem.
Part of the difficulty is that only considering binary trees affects the semantic complexity of formulas: for example, the formula $\langle a \rangle \psi \wedge \langle a \rangle \bar \psi$ (where $\bar \psi$ is the negation of $\psi$) is semantically trivial
 when interpreted on trees with only one $a$-successor while in the general case its index depends on $\psi$.
Furthermore, non-deterministic parity automata are weaker than disjunctive $L_\mu$ in the sense that some properties of binary trees can be expressed with a lower index using disjunctive form.

 \section{Deciding $\Sigma^{\mu}_2$ Reduces to a Bounding Problem}\label{sec-main}

 The first part of the proof of our main result defines a parametrised {\em $n$-challenge game} 
 on a parity game arena.
 For each finite $n$, the $n$-challenge game is described by a $\Sigma^{\mu}_2$ formula $\Psi^n$ which holds in $\mathcal{M}$ if and only if
 Even wins the $n$-challenge game on $\mathcal{M}\times \Psi$.
 We then show that a disjunctive formula $\Psi$ is equivalent to a (not necessarily disjunctive) formula in $\Sigma^{\mu}_2$ if and only if
 there is some $n$ such that $\Psi$ is equivalent to $\Psi^n$.
 As any formula can be turned into disjunctive form, this reduces the decidability of $\Sigma^{\mu}_2$ to bounding the parameter $n$.
 For the main result
 of this paper, we will only use this construction for $\Pi^{\mu}_2$ input formulas to determine equivalence to a $\Sigma^\mu_2$ formula. However, using this more general construction, a generalisation of the second part of our proof beyond $\Pi^\mu_2$ would suffice to decide $\Sigma^\mu_2$ entirely.

 When restricted to automata on binary trees and two priorities, this construction is equivalent to those found for example in \cite{colcombet2013deciding} and \cite{skrzypczakdeciding}.\\

We fix a disjunctive formula $\Psi$ with index $\{q,...,0\}$. Let $I = \{q,...,0\}$ if the maximal priority $q$ is even and $\{q+1,q,...,0\}$ otherwise.
Write $I_e$ for the even priorities in $I$.
The $n$-challenge game consists of a normal parity game augmented with a set of challenges, one for each even priority $i$. A challenge can
either be {\em open} or {\em met} and has a counter $c_i$ attached to it. 
Each counter is initialised to $n$, and decremented when the corresponding challenge is opened. The Odd player
can at any point open challenges of which the counter is non-zero, but he must do so in decreasing order: an $i$-challenge can only be opened if
every $j$-challenge for $j>i$ is opened. 
When a play encounters a priority greater or equal to $j$ while the $j$-challenge is open, the challenge is said to be met. All $i$-challenges for $i<j$ are {\em reset}. This means that
the counters $c_i$ are set back to $n$.

A play of this game is a play in a parity game, augmented with the challenge and counter configuration at each step. A play with dominant priority $d$ is winning for
Even if either $d$ is even or if every opened $d+1$ challenge is eventually met or reset.

\begin{example}
 The formula $\nu Y. \mu X. (A\wedge \Diamond X) \vee (B\wedge \Diamond Y)$  is true if on some path B always eventually holds.
 This formula does not hold in this structure:
 
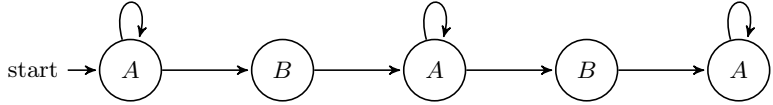
\begin{figure}
\begin{tikzpicture}[->,>=stealth',shorten >=1pt,auto,node distance=2cm,
                    semithick]
  \tikzstyle{every state}=[fill=none,draw=black,text=black]

  \node[initial,state] (A)                    {$A$};
  \node[state]         (B) [right of=A] {$B$};
  \node[state]         (C) [right of=B] {$A$};
  \node[state]         (D) [right of=C] {$B$};
  \node[state]         (E) [right of=D] {$A$};

  \path (A) edge [loop above] node {} (A)
            edge              node {} (B)
        (B) edge              node {} (C)
        (C) edge [loop above] node {} (C)
            edge              node {} (D)
        (D) edge              node {} (E)
        (E) edge [loop above] node {} (E);
\end{tikzpicture}
\end{figure}

 However, Even wins the $1$- and $2$-challenge games: her strategy is to loop in the current state until Odd opens a $2$-challenge, then meet the challenge by moving to the next
 state, as seeing a $B$ corresponds to seeing $2$. Odd will run out of challenges before reaching the
 last state. Although Odd wins the $3$-challenge game in this structure, for any $n$ it is easy to construct a similar structure
 in which he loses the $n$-challenge game but wins the parity game. This section argues that this is sufficient to show that $\nu Y. \mu X. (A\wedge \Diamond X) \vee (B\wedge \Diamond Y)$ 
 is not equivalent to any $\Sigma^{\mu}_2$ formula.
 
 In contrast, in the formula $\nu Y. \mu X. (A\wedge \Box X) \vee (B\wedge \Diamond Y)$, Odd wins the $1$-challenge game whenever he wins the parity game: he can open the challenge
 when his strategy in the parity game reaches the point at which he can avoid $B$. This formula is therefore
 equivalent to a $\Sigma^{\mu}_2$ formula, namely the alternation free formulas $\nu Y. ((A\wedge \Box Y) \vee (B\wedge \Diamond Y)) \wedge \mu X. (A\wedge \Box X) \vee B)$.
 \end{example}

\begin{definition}
 A configuration $(v,p,\bar c,r)$ of the $n$-challenge game on a parity game $G$ of index $\{q,...,0\}$ where $q$ is even consists of:
 \begin{itemize}
  \item a position $v$ in the parity game;
  \item an even priority $p$ indicating the least significant priority on which a challenge is open or $p = q+2$ if all challenges are currently met;
  \item $\bar c=(c_{0},c_{2},\ldots,c_{q})$ a collection of counter values $c_i$ for each even priority $i$.
  \item $r\in \{0,1\}$ indicating the round of the game: $1$ for Odd's turn to open challenges, $0$ for a turn in the parity game.
 \end{itemize}

At a configuration $(v,p,\bar c,1)$, corresponding to Odd's turn, he can open challenges up to any $p'\leq p$, as long as $c[i]>0$ for each $i$ such that $p'\leq i <p$.
Then the configuration becomes $(v,p',\bar c',0)$ where
$c'[i]=c[i]-1$ for all newly opened challenges $i$, that is to say $i$ such that $p'\leq i<p$ and $c'[i]=c[i]$ for all other $i$.

At the configuration $(v,p,\bar c,0)$, the player whose turn it is in the parity game decides the successor position $v'$ of $v$ and the configuration is updated to
$(v',p',\bar c',1)$ according to the priority $i$ of $v'$ as follows:
\begin{itemize}
 \item If $i\geq p$ then $p'=i+2$ if $i$ is even, $p'=i+1$ otherwise. This indicates which challenges have been met. Note that if all challenges are met, $p=q+2$.
 \item For each $j<i$, the counter value $c_j$ is reset to $n$.
 \item If $i$ is even and $c_i=0$, then the game ends immediately with a win for Even.
\end{itemize}

A play is a potentially infinite sequence of configurations starting at the initial configuration
$(v_\iota,q+2,(n,...,n),1)$, where $v_\iota$ is the initial position of the parity game. An infinite play is winning for Even if
the dominant priority on the sequence of parity game positions is $d$ but the game reaches infinitely many configurations $(v,p,\bar c,0)$ where $p>d+1$.
This is the case if $d$ is even or if all $d+1$ challenges set by Odd are either met or reset.

A strategy for Odd in a challenge game consists of two parts: a strategy which dictates when to open challenges, and a regular parity-game strategy which dictates his moves in the underlying parity game.
Even only has a parity game strategy. Both players' strategies may of course depend on the challenge configuration as well as the parity game configuration.
Given a challenge-game strategy for even $\sigma$, a challenging strategy $\gamma$ for Odd induces a normal parity game strategy $\sigma_\gamma$ for Even which does not depend on the challenge configuration.

\end{definition}

We first establish that the winning regions of the $n$-challenge games for $\Psi$ can be described by a $\Sigma^\mu_2$ formula $\Psi^n$.

\begin{lemma}\label{automaton-formula}
For all $\Psi$ and finite $n$, there is a formula $\Psi^n \in \Sigma^\mu_2$ which holds in $\mathcal{M}$ if and only if Even wins the $n$-challenge game
on $\mathcal{M}\times \Psi$.
\end{lemma}

We prove this lemma by constructing the formula $\Psi^n$.
For clarity, we will describe the alternating parity automata on labelled transition systems (see appendix \ref{app-automata}) corresponding to $\Psi^n$.
From \cite{wilke2001alternating}, this is equivalent to describing a $L_\mu$ formula.

\begin{definition}\label{automaton}
Let $A=(S,s_i,\delta,\Omega)$ be the alternating parity automaton for $\Psi$. 
We build the automata $A^n$ for $\Psi^n$ using distinct copies of $A$ for each possible challenge configuration $(p,\bar c)$.
For each even priority $p$ or $q+2$, and counter values $\bar c\in [n]^{I_e}$, the copy $A(p,\bar c)$ of $A$
corresponds to $p$ being the least significant open priority and the counter values being
$\bar c$. These components will then be combined into the automaton $A^n$.

$A(p,c)=(S^{(p,c)},s^{(p,c)}_i,\delta^{(p,c)},\Omega^{(p,c)})$ is based on $A$ using
copies $S^{(p,c)}$, $s^{(p,c)}_i$ and $S^{(p,c)}$ of
$S,s_i$ and $\delta$ respectively.
The priority function is given by $\Omega^{(p,c)}$:
\begin{itemize}
 \item If $\Omega(s)\geq p-1$ then $\Omega^{(p,c)}(s)=1$;
 \item If $\Omega(s)<p-1$ then $\Omega^{(p,c)}(s)=0$;
 \end{itemize}

The components $A(p,c)$ are linked in $A^n=(S^n,s^n_i,\delta^n,\Omega^n)$ consisting of:

\begin{itemize}
 \item The disjunct union of all component state spaces: $S^n = \biguplus_{p\in I_e,c\in [m]^{I_e}} S^{(p,c)}$;
 \item The initial state $s^n_i = s^{(q+2, \bar n)}_i$ of the component $A(q+2,\bar n)$;
 \item $\Omega^n$ defined by $\Omega^n(s)=\Omega^{(p,c)}(s)$  for $p,c$ such that $s$ is a state of the component $A(p,c)$;
 \item 
 For states $s$ in $A(p,\bar c)$ of original priority $j\geq p$, let $\delta^n(s,A)=\top$ if $c_j=0$. This corresponds to Even having met all n challenges. Otherwise,
 let $\delta^n(s,A)= s'$ such that: $s'$ is the copy of $s$ in $A(k,\bar c')$
 where $k=j+2$ if $j$ is even and $k=j+1$ otherwise, and $\bar c'[i]= n$ for $i<j$ and $\bar c'[i]=\bar c[i]$ for other $i$.
 This corresponds to the open $j$-challenge being met and all counters below $j$ being reset.

 For every state $s$ in $A(p,c)$ with original priority $j< p$, if $K$ is the set
 of even priorities smaller than $p$ such that $\bar c[k]>0$,
 let $\delta^n(s,P)$ be  $\delta^{(p,c)}(s,P) \wedge \bigwedge_{k\in K} s_k$
 where $s_k$ is the copy of $s$ in $A(k,\bar c')$, and $\bar c'[i]=\bar c[i]-1$ for $i$ such that $k\leq i<p$ and $\bar c'[i]=\bar c[i]$ otherwise.
 In other words, Odd can open challenges below $p$ if their counter-values are non-zero, by moving to the component $A(k,\bar c')$ which reflects the new challenge configuration.
 \end{itemize}

\end{definition}

\begin{proof}
The automaton described in Definition \ref{automaton} only has priorities $0$ and $1$ and therefore the corresponding formula $\Psi^n$, is in $\Sigma^{\mu}_2$. It therefore suffices to check that
this automaton indeed describes the winning regions of the challenge game.

A game in $A^n$ maps to a game in $A$, augmented with challenge configurations $(p,\bar c)$ at each state, according to the component
a state is played in. Transitions between components account for challenges being opened, met, and reset according to the rules of the game.

Let us check that $\Omega^n$ implements the winning conditions of the challenge game.
Opening challenges in $A^n$ makes the play move to lower components $A(p,c)$, as measured by $p$;
seeing high {\em original} priorities makes the play move up to higher components.
If the dominant original priority $d$ is even, then eventually the play can no longer move up to
components $A(p',c')$ with $p'>d$ from components $A(p,c)$ where $p<p'$.
Such plays eventually settle into some
$A(p,c)$ where $p> d$. Such a play is winning for Even:
it eventually only sees priority $0$. 

If $d$ is odd, then Even wins only if the play settles into some $A(p,c)$ where $p>d+1$ since those are the components in which
$d$ and lower priorities are replaced with $0$ -- this corresponds to Odd eventually not opening the challenge on $d+1$ after it has been met,
causing him to lose. If the minimum challenged priority never settles, this means the highest original priority $d$ seen infinitely often is odd
and that a $d+1$-challenge is not met -- that is to say, Odd wins the challenge game. In $A^n$ such a play is also winning for Odd since
resetting and meeting challenges corresponds to seeing priority $1$.

Therefore the automaton only accepts parity games in which Even wins the $n$-challenge game.
\end{proof}

Next
we prove our core theorem, reducing the decidability of $\Sigma^\mu_2$ to a boundedness criterion.

\begin{theorem}\label{PsiPsiM}\label{thmpsipsin}
 If a disjunctive formula $\Psi$ is semantically in $\Sigma^{\mu}_2$, then there is a finite $n$ such that  $\Psi\Leftrightarrow\Psi^n$. 
\end{theorem}

\begin{proof}

Assume that $\Psi$ is semantically in $\Sigma^{\mu}_2$, {\em i.e.} equivalent to some $\Phi$ of index $\{1,0\}$, and that for all $n$, $\Psi\nLeftrightarrow \Psi^n$. Fix $n$ to be larger than $2^{|\Psi|+|\Phi|}$. There is a structure $\mathcal{M}$,
such that Odd wins the parity game $\mathcal{M}\times \Psi$
but Even wins the $n$-challenge game on $\mathcal{M}\times \Psi$. W.l.o.g, take $\mathcal{M}$ to be finitely branching.
The overall structure of this proof is to first use a winning strategy $\tau$ for Odd in $\mathcal{M}\times \Phi$ to define a challenging strategy $\gamma$
for him in the $n$-challenge game on $\mathcal{M}\times \Psi$ (Part I). We then use Even's winning strategy $\sigma$ to add back edges to $\mathcal{M}$ (Part II), turning it into a new
structure $\mathcal{M}'$ which preserves Odd's winning strategy $\tau$
in $\mathcal{M}'\times \Phi$ while turning $\sigma_\gamma$ into a winning strategy in $\mathcal{M}'\times \Psi$ (Part III). This contradicts the equivalence of $\Phi$ and $\Psi$.

{\bf Part I.}
Let $\tau$ be Odd's winning strategy in $\mathcal{M}\times \Phi$.
Since $\mathcal{M}$ is finitely branching, for any node $v$ reachable via $\tau$, there is a finite bound $i$ such that
any play that agrees with $\tau$ sees $1$ within $i$ modal steps of any position $v\times \alpha$ that it reaches (K\"onig's Lemma).
For a branch $b$ of $\mathcal{M}$, on which $\tau$ reaches a node $v$, indicate by $\mathit{next}(b,v)$ the
$i^\mathit{th}$ node on $b$ from $v$. This node has the property that any play on the branch $b$ agreeing with $\tau$ must see a $1$ between $v$ and $\mathit{next}(b,v)$.

If $\tau$ does not agree with any plays on the branch $b$, then let $\mathit{next}(b,v)$ be a node on $b$ which $\tau$ does not reach.

Now consider the $n$-challenge game on $\mathcal{M}\times \Psi$. Let Odd's challenging strategy $\gamma$ be: to
open all challenges at the start of the game, and whenever its counter is reset;
if a challenge for a priority $i$ is met at $v$, and its counter $c_i$ is not at $0$, to open the next challenge when the play reaches a node $\mathit{next}(b,v)$ for any branch $b$,
unless the counter is reset before then ({\em i.e.} a higher priority is seen).

{\bf Part II.}
Even wins the $n$-challenge game on $\mathcal{M}\times \Psi$, so let $\sigma$ be her winning strategy. Recall that $\sigma_\gamma$ is an Even's strategy
for $\Psi$ up to the point where an $n^{\mathit{th}}$ challenge in the original challenge game is met, and undefined thereafter.
Since $\Psi$ is disjunctive, we can adjust $\mathcal{M}$ into a bisimilar structure in which the pure parity game strategy $\sigma_\gamma$ is well-behaved wherever it is defined --  it reaches
each position of $\mathcal{M}$ at either one subformula, or none.

The strategy $\sigma_\gamma$ is winning in the challenge game against any strategy for Odd which uses the challenging strategy $\gamma$. Since Odd always  eventually opens the next challenge,
the only way for him to lose is that the play reaches a position of priority $p$ when $c_p = 0$. Thus, every play is finite.

Since $\sigma_\gamma$ is well-behaved, each branch carries at most one play. For every branch $b$ the finite play it may carry must end in a long streak in which the highest priority seen is some even $p$, and it is seen at least $n$ times, corresponding to every instance of Even meeting a $p$-challenge.
As long as $n$ is sufficiently large, on every such branch there are two nodes $v$ and its descendant $w$, at which Odd opens challenges on $p$, which agree on the set of subformulas that $\sigma_\gamma$ 
reaches there in $\mathcal{M}\times \Psi$ and that $\tau$ reaches there in $\mathcal{M}\times \Phi$. We now consider the structure $\mathcal{M}'$, which is as $\mathcal{M}$ except that the predecessor of each $w$-node has an
edge to $v$ instead. The strategies $\sigma_\gamma$ and $\tau$ transfer in the obvious way to $\mathcal{M}'$.

{\bf Part III.}
We now claim that $\tau$ is winning in $\mathcal{M}'\times \Phi$ and that $\sigma_\gamma$ is winning in $\mathcal{M}' \times \Psi$. 
Starting with $\sigma_\gamma$, consider plays that do not go through back edges infinitely often. On these the dominant priority is even, as in the challenge game on $\mathcal{M}\times \Psi$.
Any play in $\mathcal{M}\times \Psi$
that agrees with $\sigma_\gamma$ which sees both $v$ and $w$ is dominated by an even priority between $v$ and $w$.
Then, as the $w$ and $v$ agree on which subformula $\sigma_\gamma$ reaches them at, an even priority dominates any play that goes through back edges in $\mathcal{M}'\times \Psi$ infinitely many times.
The strategy $\sigma_\gamma$ is therefore winning in $\mathcal{M}'\times \Psi$.

Now onto $\tau$ in $\mathcal{M}'\times \Phi$. If a branch is unchanged by the transformation, then any play on it is still winning for $\tau$, because such a play would be consistent with $\tau$ in the original game.
If a branch that $\tau$ plays on has been changed, then consider in $\mathcal{M}$ the two nodes $v$ and $w$ at which the transformation is done.
These both are nodes at which Odd opens challenges according to $\gamma$, therefore, from the definition of $\mathit{next}$ and $\gamma$, the highest priority seen between them by any play agreeing with $\tau$ is $1$.
Since $v$ and $w$ agree on which subformulas $\tau$ reaches them at, any play in $\mathcal{M}'\times \Phi$ which goes through a back-edge infinitely often sees $1$ infinitely often and is therefore winning for Odd.

This contradicts the equivalence of $\Psi$ and $\Phi$. Therefore, if $\Psi$ is semantically in $\Sigma^{\mu}_2$, then for all structures $\mathcal{M}$ the $n$-challenge game and the parity game
on $\mathcal{M}\times \Psi$ have the same winner for $n>2^{|\Phi|+|\Psi|}$.
\end{proof}

\begin{theorem}\label{thmpsipsin}
Let $\Psi\in L_\mu$, and $\Psi_d$ a disjunctive formula equivalent to $\Psi$. Then
 $\Psi$ is semantically in $\Sigma^{\mu}_2$ if and only if there is some finite $n$ such that $\Psi\Leftrightarrow\Psi_d^n$.
\end{theorem}

 \section{Deciding $\Sigma^{\mu}_2$ for $\Pi^{\mu}_2$}\label{sec-Fgame}
 
 To complete the proof of the namesake result, it suffices to show that the parameter $m$ from Theorem \ref{thmpsipsin} can be bounded. If we restrict ourselves to disjunctive $\Psi\in\Pi^\mu_2$,
we argue that the tree-building game $\mathcal{F}$ from \cite{skrzypczakdeciding} extends to arbitrary labelled transition systems and delivers such
a bound.

Since the $\mathcal{F}$ game is already well-exposed in \cite{skrzypczakdeciding}, and the adjustments to cater for disjunctive $L_\mu$ and labelled
transition systems are relatively straight-forward but verbose, the technical bulk of this section, 
that is to say the proof of Theorem \ref{restrictm1}, is left to the Appendix \ref{app-restrictm}. We obtain the following theorem.

\begin{theorem}\label{restrictm1}
 Let $\Psi\in \Pi^\mu_2$ be disjunctive. Then there is a constant $K_0$ computable from $\Psi$ such that the following statements are equivalent:
\begin{enumerate}[a)]
 \item{There is some $m$ such that $\Psi\Leftrightarrow\Psi^m$.}
 \item{$\Psi\Leftrightarrow\Psi^{K_0}$}
 \end{enumerate}
\end{theorem}

Placing everything together, we obtain our final result.

\begin{theorem}
 It is effectively decidable whether any given $\Pi^{\mu}_2$ formula is equivalent to a $\Sigma^{\mu}_2$ formula. By duality, it is also effectively decidable whether any given $\Sigma^{\mu}_2$ formula is equivalent to a $\Pi^{\mu}_2$ formula.
\end{theorem}

\begin{proof}
 Given any $\Pi^{\mu}_2$ formula $\Psi$, it can be effectively turned into a disjunctive formula $\Psi_d$ also in $\Pi^{\mu}_2$ (Lemma \ref{disjunctpitwo}).
 Then, Theorem \ref{thmpsipsin} yields that $\Psi_d$ is semantically in $\Sigma^{\mu}_2$ if and only if it is equivalent to $\Psi_d^n$ for some $n$.
 From Theorem \ref{restrictm1}, $\Psi_d\Leftrightarrow \Psi_d^n$ if and only if $\Psi_d\Leftrightarrow\Psi_d^{K_0}$ where $K_0$ is computable from $\Psi$ via $\Psi_d$. Thus, $\Psi$ is semantically in $\Sigma^\mu_2$ if and only if $\Psi\Leftrightarrow\Psi_d^{K_0}$ if and only if $\Psi_d\Leftrightarrow\Psi_d^{K_0}$.
 
 Given any $\Sigma^\mu_2$ formula, it can also be decided whether it is equivalent to a $\Pi^\mu_2$ formula, via checking whether its negation is equivalent to a $\Sigma^\mu_2$ formula.
\end{proof}

\section{Discussion} \label{sec-discussion} 

We have shown that given any $L_\mu$ formula in $\Pi^{\mu}_2$, it can be effectively decided whether it is equivalent to a $\Sigma^{\mu}_2$ formula.
This result is the $L_\mu$-theoretic counterpart of the decidability of weak definability of B\"uchi definable languages \cite{colcombet2013deciding,skrzypczakdeciding}.
The core contribution is the reduction of the decidability of $\Sigma^{\mu}_2$ for arbitrary $L_\mu$ formulas
to deciding whether the $n$-challenge game is equivalent to the model-checking parity game of a formula for any $n$. We obtain
a family of parameterised $\Sigma^{\mu}_2$ formulas $\Psi^n$ such that $\Psi$ is in $\Sigma^{\mu}_2$ if an
only if $\Psi$ is equivalent to $\Psi^n$ for some $n$.
Unfortunately, the second
part of our proof, based on \cite{skrzypczakdeciding}, is less general and only admits input formulas in $\Pi^{\mu}_2$. If this could also be generalised to
arbitrary formulas, this would yield a decidability proof for $\Sigma^{\mu}_2$.

The challenge game can be extended to constructions described by more complex $L_\mu$ formulas -- this may
turn out to be the right way to characterize higher alternation classes.
However, for Theorem \ref{PsiPsiM}, if there are more than two priorities at play, the different plays along one branch become less manageable and it is not clear how
they can inform a challenging strategy. Even when restricted to disjunctive formulas, a new technique seems to be required. 
However, the result of \cite{colcombet2008non} which achieves this for non-deterministic automata on binary trees justifies cautious optimism for the  disjunctive case.
\\

\paragraph{Achnowledgements}
We thank the anonymous reviewers for their thoughful comments and corrections.
The work presented here has been supported by an EPSRC doctoral studentship at the University of Edinburgh.

\bibliographystyle{splncs}
\bibliography{Cie.bib}

\appendix

Appendix \ref{app-automata} defines our automata notation.
Appendix \ref{app-preserve-pi} proves that the transformation into disjunctive from turns $\Pi^\mu_2$ formulas into disjunctive $\Pi^\mu_2$ formulas.
Finally, Appendix \ref{app-restrictm} contains in detail the second part of our main decidability proof, scetched in  Section \ref{sec-Fgame}.

 \section{$L_\mu$ automata}\label{app-automata}
 
 $L_\mu$ formulas are known to be equivalent to alternating parity automata \cite{wilke2001alternating} on labelled transition systems. 
 For notational purposes, we recap the definition of alternating parity automata which match the syntax
of $L_\mu$ used in this paper.
 \begin{definition}
  An alternating parity automaton is a tuple $A=(Q,q_i,\delta,\Omega)$ where
  \begin{itemize}
   \item $Q$ is a finite set of states;
   \item $q_i\in Q$ is the initial state;
   \item $\delta$ is a transition function which maps each state to a {\em transition condition} as defined below;
   \item $\Omega:Q\rightarrow I$ is a priority function.
  \end{itemize}
  
  A transition condition is:
  \begin{itemize}
   \item $\top,\bot$;
   \item $P,\neg P$ for $P\in \mathit{Prop}$;
   \item $q,[a]q, \langle a\rangle q$ for $q\in Q$ and $a\in \mathit{Act}$;
   \item $q\vee q',q\wedge q'$ for $q,q'\in Q$.
  \end{itemize}
  
We assume the standard accepting condition and correspondence with $L_\mu$ formulas on labelled transition systems. For further details, refer to \cite{wilke2001alternating}.
 \end{definition}

In contrast, a non-deterministic automaton on labelled binary trees is a particular alternating parity automaton with a transition function restricted to $\delta(s)=
\bigvee_i\bigwedge_j(q_{ij},a_{ij})$ where for each $i$ the $a_{ij}$ are pairwise different. Note that since there is only one successor for each label,
$(q_{ij},a_{ij})$ is equivalent to both  $[a_{ij}]q_{ij}$ and  $\langle a_{ij} \rangle q_{ij}$.

\section{Disjunctive form preserves $\Pi^{\mu}_2$}\label{app-preserve-pi}

Recall that Lemma \ref{disjunctpitwo} states that given a $\Pi^{\mu}_2$ formula, the transformation into disjunctive form from \cite{Walukiewicz2000} yields a disjunctive $\Pi^{\mu}_2$
 formula.

\begin{proof}
 For full details of the transformation, see \cite{Walukiewicz2000}. Given this transformation of a formula into a tableau-equivalent disjunctive formula,
 the argument that it preserves $\Pi^{\mu}_2$ is relatively simple. 
  Assume that a disjunctive formula $\Psi$ is not in $\Pi^{\mu}_2$ and therefore
 has a cycle with a $\mu$-bound fixpoint $X$ which is more significant than a $\nu$-bound fixpoint $Y$.
 Consider a tableau-equivalent formula $\Phi$. 
 
 Since the tableaus of $\Psi$ and $\Phi$ are tableau equivalent, they agree on the parity of infinite paths. In particular, any cycle dominated by a $\nu$-bound $Y$ according to the tableau of $\Psi$
 is even, {\em i.e.}, has no $\mu$-traces in either tableau, while all paths dominated by a $\mu$-bound $X$ according to the tableau of $\Psi$ have a $\mu$-trace in both tableaus.
 We consider a cycle  $\pi_X$  which, in the tableau of $\Psi$, is dominated by a $\mu$-bound $X$ but sees a cycle dominated by a $\nu$-bound $Y$ many times -- say $n$ times -- without 
 seeing a more significant fixpoint
 in between. Choose such a path for $n$ larger than the largest label in the tableau of $\Phi$. This ascertains that any trace, in particular any $\mu$-trace in the tableau of $\Phi$,
 on $\pi_X$ reaches the node labelled by $Y$
 twice at the same subformula $\alpha$ while it goes through the $Y$-cycle $n$ times. 
 The highest priority $p$ such a trace sees between the two instances of $\alpha$
 has to be even, since the $Y$-cycle is even. However if the trace is a $\mu$-trace, as at least one trace on $\pi_X$ must be in the tableau of $\Phi$, then it has to be dominated
 by an odd priority that must therefore be more significant than $p$. Such traces do not exist in the tableaus of $\Pi^{\mu}_2$ formulas, so $\Phi$ is not $\Pi^{\mu}_2$.
 
 Therefore, any $\Pi^{\mu}_2$ formula is only tableau equivalent to disjunctive $\Pi^{\mu}_2$ formulas.
\end{proof}

\section{Proof of Theorem \ref{restrictm1}}\label{app-restrictm}

This section aims to prove that given a disjunctive
 $\Psi\in \Pi^\mu_2$, there is a constant $K_0$ computable from $\Psi$ such that the following statements are equivalent:
\begin{enumerate}[a)]
 \item{There is some $m$ such that $\Psi\Leftrightarrow\Psi^m$.}
 \item{$\Psi\Leftrightarrow\Psi^{K_0}$}
 \end{enumerate}
 
 To do so, we extend the $\mathcal{F}$-game introduced in \cite{skrzypczakdeciding} to the $L_\mu$ setting and show that Even wins $\mathcal{F}(n)$ if and only if
 $\Psi \nLeftrightarrow \Psi^n$. The proof of decidability for $\mathcal{F}$ follows 
 \cite{skrzypczakdeciding} closely.
 
Let $\Psi$ be a formula in disjunctive form, of index $I = \{2,1\}$.
Fix $\bar \Psi$, the negation of $\Psi$ in disjunctive form (it may have a different index from $\Psi$).
We defined the challenge game for arbitrary $L_\mu$ formulas; however, when restricted to $\Pi^{\mu}_2$,
there is only one challenge. 
A binary state $\{\mathit{open},\mathit{met}\}$ and one counter suffice to represent the challenge configuration.

\begin{definition}
A position $(S,\phi,\kappa,r)$ of $\mathcal{F}(\beta)$ for $\beta\in \omega+1$ (natural numbers and $\omega$) is:
\begin{itemize}
 \item $S$ a set of {\em active states} consisting of pairs $(f,p)$ where $f\in \mathit{sf}(\Psi)$ and $p\in \{\mathit{open,met}\}$.
 \item $\phi\in \mathit{sf}(\bar \Psi)$;
 \item $\kappa:S\rightarrow (\beta+1)$ a function that assigns to each active state a counter value.
 \item $r\in \{0,1\}$ a sub-round number.
\end{itemize}
\end{definition}
The initial position is $(\{(\Psi,q+2)\},\bar\Psi,\kappa,0)$ where $\kappa(\Psi,q+2)= \beta $. Then, following \cite{skrzypczakdeciding}, we define multi-transitions.

\begin{definition}
A multi-transition from a position $(S,\phi,\kappa,r)$ to $(S',\phi',\kappa',r')$ consists of:
\begin{itemize}
 \item The pre- and post-states $(S,\phi,\kappa,r)$ and $(S',\phi',\kappa',r')$ where $r'=r+1\; \mathtt{ mod }\; 2$;
 \item a set $e$ of edges from the active states in $S$ to the active states in $S'$;
 \item a set $\bar e\subseteq e$ of boldfaced edges, where exactly one ends at each $(f,p)\in S'$
\end{itemize}
\end{definition}

The intention of the game $\mathcal F(n)$ is to let Even win iff there is a model for $\bar\Psi\wedge\Psi^n$.
The positions can be seen as attempts to build a branch of such a model, including witnesses for $\bar\Psi$ and tracking potential
opened challenges for $\Psi^n$. 
During round 0, the Odd player can restrict what challenges he may open.
Then, in round 1, Even decides on the propositional variables true in the current state and a finite set of successor states.
She also extends her strategies on $\Psi$ and $\bar \Psi$ to those successors.
Odd then chooses a successor, which induces a new set of active states. Since the same active state may be reached from more than one predecessor state,
he also specifies boldfaced edges to each new active state.
The challenge-configuration is updated to reflect any challenges met or reset on the traces along boldfaced edges.\\

More formally, if the current configuration is $(S,\phi,\kappa,r)$, then the players construct a multi-transition to a new
configuration in the following ways:
\begin{itemize}
 \item (R0) $r=0$. Odd chooses a set $C$ of pairs $(f,\mathit{open})$ such that $(f,\mathit{met})\in S$, and
 $\kappa(f,p)>0$. 
 The new active states are then $S'=S\cup C$ and Odd must specify with a bold-faced edge a predecessor $(f,p)\in S$ for each $(f,p')\in S'$. 
 This predecessor must satisfy $\kappa(f,p)>0$.
 For each such new state $(f,\mathit{open})$ with predecessor $(f,\mathit{met})$,
 set $\kappa'(f,p')=\kappa(f,p)-1$ if $\beta\not=\omega$; if $\beta=\omega$ then $\kappa$ is always the constant $\omega$.
 The new configuration is $(S',\phi,\kappa',1)$.
 
 \item (R1) $r=1$. Even chooses:
 
 (i) a set of propositional variables $P$, 
 
 (ii) a set of successors $N_a = \{s_0,...,s_n\}$ for each action $s\in \mathit{Act}$ no larger than $|\Psi|+|\bar \Psi|$,
 
 (iii) a next Odd-choice formula $A_\phi \wedge \bigwedge_{a\in \mathit{Act}}  {\xrightarrow{a}} B_a$ of $\phi$ where $A_\phi$ respects $P$,
 
 (iv) a surjection $g_{{\xrightarrow{a}} B_a} : N_a \rightarrow B_a$ for each $a\in \mathit{Act}$,
 
 (v) a set $D$ consisting of a pair $(f',p')$ for every $(f,p)\in S$ where $f'$ is a next Odd-choice formula $f'=A_f \wedge \bigwedge_{a\in \mathit{Act}} {\xrightarrow{a}} B_a$
 such that $A_f$ respects $P$, and if the trace from $f$ to $f'$ sees $2$, then $p'=\mathit{met}$, otherwise $p'=p$. 
 
 (vi) for each chosen $(f',p')$, where $f'=A_f \wedge \bigwedge_{a\in \mathit{Act}} {\xrightarrow{a}} B_a$ a surjection $g_{(f',p'),a} : N_a \rightarrow B_a$ for each $a\in \mathit{Act}$.

 Odd responds by choosing $s'$ out of the successors. This induces a new set of active states: if $s'$ is an $a$-successor,
 the set $S'$ consisting of $(g_{(f,p),a}(s'),p')$ for each $(f,p)\in D$ such that if $2$ is seen
 on the trace from $f$ to $g_{(f,p),a}(s')$ then $p'=\mathit{met}$, else $p'=p$.
 The edges are $((f,p),(f'',p'')$ such that Even chooses $(f',p')$ at (v) from $(f,p)$, $g_{(f',p'),a}(s')=f''$, and $p''$ if $\mathit{met}$ or $p'$ accorsing to whether a $\nu$-bound variable is seen
 in between.
Finally, Odd also chooses for each $(f',p')\in S'$ a predecessor $(f,p)\in S$ to make the edge $((f,p),(f',p'))$ bold-faced.
The new configuration is then $(S',g_{{\xrightarrow{a}} B_a}(s'),\kappa',0)$.
 
\end{itemize}

A play is an infinite sequence of game configurations, linked by multi-transitions. These multi-transitions differ slightly from \cite{skrzypczakdeciding} to reflect the different context.
A play carries one $\bar \Psi$-trace and one or several $\Psi$ challenge traces, some of which are bold-faced.
Even wins a play if 
\begin{enumerate}[a)]
\item for every infinite $\Psi$-challenge trace Even meets every challenge opened by Odd, and
\item at least one of the following is true: 
\begin{enumerate}[I)]
\item on some boldfaced trace, infinitely many challenges are opened and met, or
\item the $\bar \Psi$ trace is winning for Even.
\end{enumerate}
\end{enumerate}

If Even wins $\mathcal{F}(n)$ with conditions a) and b II), this will give rise to a model of $\bar\Psi\wedge\Psi^n$. 
Conditions a) and b I) can only be satisfied in $\mathcal F^\omega$ and serve as tool to establish who will win $\mathcal{F}(n)$ for large n.

\begin{lemma}\label{FNandPsiN}
 For finite $n$, Even wins the $\mathcal{F}(n)$ game if and only if $\Psi\nLeftrightarrow \Psi^n$.
\end{lemma}

\begin{proof}
 First assume that Even wins $\mathcal{F}(n)$ for some finite $n$.
 We consider the following family of strategies for Odd in $\mathcal{F}(n)$: at round 0, he chooses $C$ to include every  pair  $(f,p')$ such that $(f,p)\in S$ and $\kappa(f,p)>0$.
 That is to say, he allows himself to set challenges whenever the counter values permit it.
 This means that in the structure we build, Even will have to have a strategy against all possible challenging strategies.
He also chooses a bold-faced edge inducing the largest $\kappa$.

 This defines Odd's strategy apart from the choice of direction. Such a partial strategy, combined with Even's winning strategy $\sigma$ in $\mathcal{F}(n)$ induces a structure $\mathcal{M}$.
 Even's winning strategy $\sigma$ in $\mathcal{F}(n)$ induces a strategy in the $n$-challenge game $\mathcal{M}\times \Psi$. 
From the winning condition a) of $\mathcal{F}(n)$, every play in $\mathcal{M}\times \Psi^n$
 which agrees with this strategy must be winning for Even. On the other hand, Odd's challenging strategy does not open infinitely many challenges
 on any boldfaced trace. Thus, b II) holds, and Even's strategy in $\mathcal{M}\times \bar \Psi$ induced by $\sigma$ must be winning. Odd therefore wins $\mathcal{M}\times \Psi$.
 
 For the other direction assume that there is a structure $\mathcal{M}$ such that  Odd wins $\mathcal{M}\times \Psi$ but Even wins the $n$-challenge game on the same arena with a strategy $\sigma$.
 Let $\mathcal{M}$ and $\sigma$ be such that for each Even's strategies $\sigma_\gamma$, for all challenging strategies $\gamma$,
 in  $\mathcal{M}\times \Psi$ the strategy
 only agrees with one play per branch; let the same be true for the winning strategy $\bar \sigma$ in $\mathcal{M}\times \bar \Psi$. This is possible due to both $\Psi$ and $\bar \Psi$ being in disjunctive form.
 
 Her strategy in $\mathcal{F}(n)$ is to build $\mathcal{M}$. At each game configuration, she keeps track of
 the state $v$ in $\mathcal{M}$ that it corresponds to. At $(S,\phi, \kappa, 1)$ she then plays:
 \begin{itemize}
  \item the set of propositional variables at $v$ and the sets  $N_a$ of $a$-successors of $v$;
  \item the set $D$ of pairs $(f',p')$ for each $(f,p)\in S$ such that $f'$ is the next Odd-choice formula of $f$ which 
her winning strategy plays at $v\times f$ if the current challenging configuration is $p$ with counter value $\kappa(f,p)$; $p'=p$ if no $2$ is seen along these steps, otherwise $p'=\mathit{met}$;

\item for each $(f, p)\in D$ and $a\in \mathit{Act}$, the surjection $g_{(f',p),a}$ which map each $a$-successor $v'$ to the unique immediate subformula $b$ of $f$
such that $\sigma$ plays $v'\times b$ from $v \times f$
when the challenge configuration is $p$ with counter $\kappa(f, p)$
 \item the next Odd-choice formula $\phi'$ her winning strategy $\bar\sigma$ in $\mathcal{M}\times \bar \Psi$ plays at $v\times \phi$;

   \item for the Odd-choice formula $\phi'$  of $\bar \Psi$, the surjections $g_{\phi,a}$ which map each $a$-successor $v'$ to the unique subformula $\phi''$
   such that $\bar\sigma$ plays $v'\times \phi''$ from $v \times \phi'$.
 \end{itemize}

The $\Psi$-traces in any play that agrees with this strategy correspond to plays agreeing with $\sigma$, which guarantees that the $\mathcal{F}(n)$-play
satisfies the winning condition (a). The $\bar \Psi$ trace corresponds to a play that agrees with $\bar \sigma$, satisfying the winning condition (b II).

\end{proof}

It then remains to be shown that Even wins $\mathcal{F}(\omega)$ if and only if she wins $\mathcal{F}(n)$ for all $n$ and that the winner of $\mathcal{F}(\omega)$ is decidable.

\begin{lemma}\label{boundingN}
  There is a finite value $K_0$, computable from $\Psi$ and $\bar \Psi$, such that 
  Even wins $\mathcal{F}(\omega)$ if and only if she wins $\mathcal{F}({K_0})$.
\end{lemma}

\begin{proof}
First, note that if Odd wins $\mathcal{F}({\omega})$, he can win with a strategy
which in round $0$ chooses $C$ to be the empty set whenever $S$ contains an active state $(f,\mathit{open})$.
In other words, Odd can always wait that all traces meet opened challenges to open a new challenge. Let such a strategy be called {\em patient}.
Note that $\mathcal{F}(\omega)$ is a finite game with a regular winning condition.
 Its winner therefore has a finite memory winning strategy. Suppose that Odd wins the game. Let $M$ be the size of
 the memory of Odd's patient winning finite memory strategy $\tau$ in $\mathcal{F}(\omega)$.
 Let $K_0$ be the product of $M$, the number of configurations of $\mathcal{F}(\omega)$ and the set of possible active states.
 
 We argue that $\tau$ is a winning strategy in $\mathcal{F}({K_0})$. First,
 we have to show that it is a valid strategy, i.e., Odd never tries to open a challenge with an empty counter.
 This could only occur if some bold-faced trace opened $K_0$ challenges.
  If that was the case then, $K_0$ being very large and Odd's memory being only $M$, there would be a
looping fragment 
 along this play  in which on a boldfaced trace a challenge is both opened and met. Following this boldfaced trace, a challenge would be open and met infinitely often.
 Furthermore, since $\tau$ is patient, winning condition $(a)$ would also hold, since Odd only opens challenges when
 all traces are in the $\mathit{met}$ state. This contradicts the assumption that $\tau$ is winning in $\mathcal{F}(\omega)$.
 
 We then argue that if Odd plays using $\tau$, this is a winning strategy for Odd in $\mathcal{F}({K_0})$. 
  Counting challenges does not affect the first winning condition whereby if Even is to win, in every infinite trace, Even must meet every challenge.
 So if a play that agrees with $\tau$ is winning for Odd in $\mathcal{F}(\omega)$ because on some trace Even fails to meet some challenge, then
 the same is true in $\mathcal{F}({K_0})$.
 Furthermore, as argued above, $\tau$ does not open more that $K_0$ challenges, so in no play does condition (b I) hold. Finally, condition (b II)
 is not affected by the addition of counters and inherits the winner from the $\mathcal{F}(\omega)$ game.
 As a result, $\tau$ is winning in $\mathcal{F}({K_0})$.
 
In the case that Even wins the  $\mathcal{F}(\omega)$ game, note that a winning strategy for Even in $\mathcal{F}(\omega)$ is a 
winning strategy in any $\mathcal{F}(n)$ for finite $n$. In particular, she would win $\mathcal{F}({K_0})$.
\end{proof}

\begin{corollary}\label{FomegaFn}
 Even wins $\mathcal{F}(\omega)$ if and only if she wins $\mathcal{F}(n)$ for all $n$.
\end{corollary}

\setcounter{theorem}{2}
\begin{theorem}
 Let $\Psi\in \Pi^\mu_2$ be disjunctive. Then there is a constant $K_0$ computable from $\Psi$ such that the following statements are equivalent:
\begin{enumerate}[a)]
 \item{There is some $m$ such that $\Psi\Leftrightarrow\Psi^m$.}
 \item{$\Psi\Leftrightarrow\Psi^{K_0}$}
 \end{enumerate}
\end{theorem}

\begin{proof}
We only need to show that the first statement implies the second one.
 Let $\Psi$  be a disjunctive formula in $\Pi^{\mu}_2$ such that the first statement holds true.  Lemma \ref{FNandPsiN} yields that $\Psi\nLeftrightarrow \Psi^n$ if and only if Even wins $\mathcal{F}(n)$.  From Lemma \ref{boundingN} and Corollary \ref{FomegaFn}, this is true for all n if and only if she wins $\mathcal{F}({K_0})$, where $K_0$ depends only on $\Psi$. However, the first statement implies that Even loses $\mathcal{F}({K_0})$.
 Thus, Odd wins $\mathcal{F}({K_0})$, implying $\Psi \Leftrightarrow \Psi^{K_0}$ by Lemma \ref{FNandPsiN}.
\end{proof}

\end{document}